\documentclass[envcountsame,envcountsect,a4paper]{llncs}

\usepackage[english]{babel}
\usepackage{amsmath}
\usepackage{array}
\usepackage[boxed,french]{algorithm2e}
\usepackage{amssymb}
\usepackage{mathrsfs}
\usepackage{float}
\usepackage[shortlabels]{enumitem} %
\usepackage{pdfpages}
\usepackage{url}
\usepackage[linkcolor = blue,colorlinks = false]{hyperref} 

\newcommand*{\EnQuR}[2]%
{\ensuremath{%
    #1/\!\raisebox{-.65ex}{\ensuremath{\mathcal{#2}}}}}

\newcommand{\NP}{\ensuremath{\textsf{NP}}}

\newcommand{\NPcNP}{\ensuremath{\textsf{NP} \cap \textsf{coNP}}}
\newcommand{\PT}{\ensuremath{\textsf{P}}}
\newcommand{\FP}{\ensuremath{\textsf{FP}}}

\newcommand{\TFNP}{\ensuremath{\textsf{TFNP}}}
\newcommand{\QBF}{\ensuremath{\textrm{-QBF}}}

\newcommand{\SEN}{\mathrm{SEN}}
\newcommand{\SAT}{\mathrm{SAT}}
\newcommand{\MOD}{\mathrm{MOD}}
\newcommand{\STRUC}{\mathrm{STRUC}}
\newcommand{\GRAPH}{\mathrm{GRAPH}}
\newcommand{\Can}{\mathrm{Can}}

\title{Relativization of Gurevich's Conjectures}

\author{Anatole Dahan\inst{1} \and Anuj Dawar\inst{2}}
\institute{ENS Paris \and University of Cambridge 
}

\begin{document}

\maketitle

\begin{abstract}
Gurevich (1988) conjectured that there is no logic for $\PT$ or for
$\NPcNP$.  For the latter complexity class, he also showed that the
existence of a logic would imply that $\NPcNP$ has a complete problem
under polynomial time reductions.  We show that there is an oracle
with respect to which $\PT$ does have a logic and $\PT \neq \NP$.  We
also show that a logic for  $\NPcNP$ follows from the existence of a
complete problem and a further assumption about canonical labelling.
For intersection classes $\Sigma_n^p \cap \Pi_n^p$ higher in the
polynomial hierarchy, the existence of a logic is equivalent to the
existence of complete problems.
\end{abstract}

\section{Introduction}\label{sec:intro}
In a highly influential paper published in 1988~\cite{gurevich1}, Yuri
Gurevich put forth the conjecture that there is no logic that captures
polynomial time computation.  The question of whether there is a logic
for $\PT$ has been a major driver of research in finite model theory
and descriptive complexity in the last thirty years.  In this line of
work, the exact formulation of the question given by Gurevich has
played a central role. Roughly speaking (a precise definition is
given later), the question is whether there is a recursive set $S$ of
\emph{polynomially-clocked} deterministic Turing machines each of
which decides an isomorphism-closed class of structures and such that
for every such class in $\PT$, there is a machine in $S$ witnessing
this fact.

Gurevich's conjecture that there is no logic for $\PT$ implies that
$\PT$ is different from $\NP$.  This is not, as is often assumed, a
simple consequence of Fagin's result~\cite{Fagin} that there is a
logic for $\NP$, i.e.\ existential second-order logic.  Indeed,
knowing Fagin's theorem and assuming $\PT=\NP$ does not immediately
yield a \emph{computable} translation from sentences of existential
second-order logic to \emph{deterministic} polynomially-clocked
machines.  The argument requires a little bit more work.  There is,
however, another argument that takes us from $\PT=\NP$ to a refutation
of Gurevich's conjecture.  This relies on the fact that  $\PT=\NP$ would imply the
collapse of the polynomial hierarchy and, in particular, that there is a
polynomial-time algorithm for producing a canonical labelling of a
graph (see~\cite{BlassGurevich-equivalence}).  A polynomial-time
algorithm for canonical labelling of graphs yields a logic for $\PT$
(see~\cite[Proposition~1.7]{Dawar-thesis}).  Indeed, much of the
research around the existence of logics for $\PT$ has been concerned
with the existence of canonical labelling algorithms on suitable
classes of structures.

Thus, while $\PT =\NP$ would imply the refutation of Gurevich's
conjecture, the converse of this statement is not known.  Indeed, it
is often said that it is entirely consistent with our knowledge that
$\PT$ is different from $\NP$ but there is a logic for $\PT$.  The
second author of the present paper made this statement in a lecture in
2012 and was challenged from the audience to provide evidence for it.
Theodore Slaman asked if there is a relativized world in which $\PT$
is different from $\NP$ but there is a logic for $\PT$.  In
Section~\ref{sec:relativization} we show that this is, indeed, the case.  That is we give
a construction of an oracle $A$ such that there is a logic for
$\PT^A$, but $\PT^A \neq \NP^A$.  This should be contrasted with the
result shown in~\cite{Dawar} that if $\PT = \NP$ (in the unrelativized
sense), then there is a
logic for $\PT^A$ for all sets $A$.

Gurevich also conjectured in~\cite{gurevich1} that there is no logic
for the complexity class $\NPcNP$.  Relativizations of this
conjecture were considered in~\cite{Dawar} (published on the occasion
of Yuri's 70th birthday) where it was shown that
this conjecture is subject to the relativization barrier, in the sense
that there are relativized worlds in which it is true and also
relativized worlds in which it is false.  The construction of an
oracle for which $\NPcNP$ does not have a logic is based on
known constructions of oracles for which $\NPcNP$ does not
admit complete problems under polynomial-time reductions (see~\cite{Sip82}), and the fact
that a logic for $\NPcNP$ would imply the existence of complete
problems even under first-order reductions.  This last statement is a
theorem stated in~\cite[Theorem~4]{Dawar} though the proof was omitted as it is
similar to the well-known proof of the corresponding statement for
$\PT$~\cite{Dawar-jlc}.  In Section~\ref{sec:intersection}, we give a proof of this
fact as a special case of a more general result about $\Delta$-levels
of the polynomial hierarchy.  We are able to show, in Section~\ref{sec:reductions}, for all levels above
the first that the existence of complete problems under
polynomial-reductions is \emph{equivalent} to the existence of
complete problems under first-order reductions.

\section{Preliminaries}\label{sec:prelim}
We work with finite relational signatures.  We write $\sigma$ for an
arbitrary such signature.  All our structures are finite, so a
$\sigma$-structure is a finite set $A$ along with an interpretation on
$A$ of every relation symbol in $\sigma$.  We write
$\STRUC[\sigma]$ to denote the collection of all finite
$\sigma$-structures.  We do not consider any specific signatures
except that of graphs, i.e.\ where $\sigma$ consists of the single
binary relation $E$.  We refer to this signature as $\GRAPH$.  We
assume a standard encoding of finite relational structures as strings,
as given in~\cite{gurevich1}.  We write $|S|$ for the size (i.e.\
number of elements) of a structure $S$, which is related by a
polynomial factor to the length of the string encoding $S$.  As these
polynomial factors are unimportant for our discussion, we do not
distinguish between $S$ and the string encoding it.  Note that,
strictly speaking, an encoding of $S$ depends on $S$ \emph{and} a
choice of order on the universe of $S$.  Where this is significant, we
mention the order explicitly.  For full background material on finite
model theory, the reader is referred to~\cite{EbbinghausFlum}.

We begin by stating the definition of a logic
given by Gurevich~\cite{gurevich1}

\begin{definition}[Logic]
A logic $\mathcal L$ is a pair $(\SEN,\SAT)$ of functions, taking a
signature $\sigma$ as parameter, such that
\begin{itemize}
\item $\SEN(\sigma)$ is a recursive set. We call $\varphi\in \SEN(\sigma)$ an $\mathcal L$-sentence on $\sigma$.
\item $\SAT(\sigma)$ is a recursive subset of $\{(S,\varphi) \mid
  \varphi\in \SEN(\sigma), S \in \STRUC[\sigma] \}$, such that for two isomorphic structures $S$ and $S'$ \[\forall \varphi\in \SEN(\sigma), (S,\varphi)\in \SAT(\sigma) \iff (S',\varphi)\in \SAT(\sigma)\]
\end{itemize}
If $\varphi$ is an $\mathcal L$-sentence on $\sigma$, we write
$\MOD[\varphi]$ to mean $\{S \mid (S,\varphi)\in \SAT(\sigma)\}$.
\end{definition}

Next, we reproduce Gurevich's definition of a logic capturing
polyonomial time. 
\begin{definition}\label{def:Pcapture}
A logic $\mathcal L$ captures $\PT$ if:
\begin{itemize}
\item there is a Turing machine $\mathcal C$ such that, on every input
  $\mathcal L$-sentence $\varphi$ of signature $\sigma$, $\mathcal C$
  outputs a pair $(M,p)$, where $M$ is a deterministic Turing machine
  and $p$ is a polynomial, such that for all $\sigma$-structures $S$,
  $S\in \MOD[\varphi]$ if, and only if, $M$ accepts $S$ within time
  $p(|S|)$; and
\item if $\mathcal P\subseteq \STRUC[\sigma]$ is an isomorphism-closed
  class of structures that belongs to $\PT$, then there exists an
  $\mathcal L$-sentence $\varphi$ of signature $\sigma$ such that
  $\MOD[\varphi] = \mathcal P$. 
\end{itemize}
\end{definition}

Definition~\ref{def:Pcapture} formalises the definition from the
opening paragraph of Section~\ref{sec:intro}.  It does not give a
general definition of capturing a logic for a complexity class, as it
crucially depends on the idea of membership of a class of structures
in $\PT$ being \emph{witnessed} by a pair $(M,p)$.  Different
complexity classes have rather different notions of witness.  In this
spirit, the following is Gurevich's definition of a logic capturing
$\NPcNP$.

\begin{definition}\label{def:NPcNPcapture}
A logic $\mathcal L$ captures $\NPcNP$ if :
\begin{itemize}
\item There is a Turing machine $\mathcal C$, such that, on every
  input $\mathcal L$-sentence $\varphi$ of signature $\sigma$,
  $\mathcal C$ outputs a triple $(M,N,p)$ where $M$ and $N$ are
  non-determinisitic Turing machines and $p$ is a polynomial such that :
\begin{itemize}
\item $\forall S\in \STRUC[\sigma], S\in \MOD[\varphi]$ if, and only
  if, there is a computation of $M$ of length at most $p(|S|)$ by which $M$ accepts $S$.
\item $\forall S\in \STRUC[\sigma], S\in \MOD[\varphi]$ if, and only
  if, all computations of $N$ on input $S$ of length at most $p(|S|)$
  lead to acceptance.
\end{itemize}
\item If $\mathcal P\subseteq \STRUC[\sigma]$ is an isomorphism-closed
  class of structures that belongs to $\NPcNP$, then there exists an
  $\mathcal L$-sentence $\varphi$ of signature $\sigma$ such that
  $\MOD[\varphi] = \mathcal P$. 
\end{itemize}
\end{definition}

Here the witness to membership in the class $\NPcNP$ is given by a
triple $(M,N,p)$.  It should be noted that in the case of
Definition~\ref{def:Pcapture}, the collection of witnesses $(M,p)$ is
a recursive set where we put a \emph{semantic}, undecidable condition
that the class of structures accepted by $(M,p)$ is isomorphism-closed.  In contrast, in the case of
Definition~\ref{def:NPcNPcapture}, we have \emph{two} separate
semantic conditions, namely that the two machines in the witness agree
on the class of structures accepted \emph{and} that this class is
isomorphism-closed.  As noted in~\cite{Dawar}, it is the first of
these conditions that means that $\NPcNP$ is not even known to have
complete problems under polynomial-time reductions and that Gurevich's
conjecture with regard to $\NPcNP$ is subject to the relativization
barrier.  

It was proved in~\cite{Dawar-jlc} that there is a logic for $\PT$ in
the sense of Defintion~\ref{def:Pcapture} if, and only if, there is a
problem in $\PT$ that is complete under first-order reductions.  A
similar statement for a logic for $\NPcNP$ was stated
in~\cite{Dawar}.  In the present paper, we prove this, and extend it
to higher levels of the polynomial hierarchy.  First, we introduce the
relevant definitions and notations in connection with the polynomial
hierarchy. 

For any set $A$, $\PT^A$ denotes the class of languages which are
accepted by some deterministic Turing machine with an oracle for $A$
in polynomial time.  Similarly $\NP^A$ denotes the class of languages
which are accepted by some nondeterministic Turing machine with an
oracle for $A$ in polynomial time.  The classes of the polynomial
hierarchy are defined as follows.
\begin{definition}\label{def:PH}
For all $n \geq 1$,
  \begin{itemize}
  \item   A language $L$ is in $\Sigma_1^p$ if, and only if, $L \in
    \NP$.
  \item  A language $L$ is in $\Sigma_{n+1}^p$ if,
    and only if, there is some $A \in \Sigma_n^p$ such that $L \in
    \NP^A$. 
  \item  A language $L$ is in $\Pi_n^p$ if, and only if, $\bar{L} \in
    \Sigma_n^p$.
  \item A language $L$ is in $\Delta_{n+1}^p$ if, and only if, there is some $A \in \Sigma_n^p$ such that $L \in
    \PT^A$. 
  \end{itemize}
\end{definition}
It is clear that $\Delta_n^p \subseteq \Sigma_n^p \cap \Pi_n^p$ for
all $n$, but equality is not known for any $n$.  In terms of the
existence of a logic, we know by Fagin's theorem~\cite{Fagin} that
there is a logic for $\NP$, and this is extended by~\cite{Stockmeyer}
to show that for each $n$, $\Sigma_n^p$ is captured by the $\Sigma_n$-fragment of
second-order logic.  Similarly, $\Pi_n^p$ is captured by the
$\Pi_n$-fragment.  We do not, however, obtain by these means a logic
for  $\Sigma_n^p\cap\Pi_n^p$.  To make this precise, we introduce here
a definition of what it would mean to capture these classes (in the
spirit of Definition~\ref{def:NPcNPcapture}).  Before doing
so, it is useful to recall that we have, for each $n$, a problem that
is complete for $\Sigma_n^p$ under polynomial-time reductions.  For
our purposes, it suffices to take one such problem, $\Sigma_n\QBF$.
This is the problem of deciding the truth of a quantified Boolean
formula in prenex form with $n-1$ alternations of quantifiers,
starting with an existential block.  By the fact that this problem is
$\Sigma_n^p$-complete, it follows that $\NP^{\Sigma_n\QBF} =
\Sigma_{n+1}^p$ for all $n$.  

\begin{definition}\label{def:PHcapture}
For any $n \geq 1$, a logic $\mathcal L$ captures  $\Sigma_{n+1}^p\cap\Pi_{n+1}^p$ if :
\begin{itemize}
\item There is a Turing machine $\mathcal C$, such that, on every
  input $\mathcal L$-sentence $\varphi$ of signature $\sigma$,
  $\mathcal C$ outputs a triple $(M,N,p)$ where $M$ and $N$ are
  non-determinisitic oracle Turing machines and $p$ is a polynomial such that :
\begin{itemize}
\item $\forall S\in \STRUC[\sigma], S\in \MOD[\varphi]$ if, and only
  if, there is a computation of $M$ with oracle $\Sigma_n\QBF$ of
  length at most $p(|S|)$ by which $M$ accepts $S$.
\item $\forall S\in \STRUC[\sigma], S\in \MOD[\varphi]$ if, and only
  if, all computations of $N$ with oracle $\Sigma_n\QBF$ on input $S$ of length at most $p(|S|)$
  lead to acceptance.
\end{itemize}
\item If $\mathcal P\subseteq \STRUC[\sigma]$ is an isomorphism-closed
  class of structures that belongs to  $\Sigma_{n+1}^p\cap\Pi_{n+1}^p$, then there exists an
  $\mathcal L$-sentence $\varphi$ of signature $\sigma$ such that
  $\MOD[\varphi] = \mathcal P$. 
\end{itemize}
\end{definition}

\section{Capturing intersection classes in the polynomial hierarchy}
The relationship between the existence of a logic for a complexity
class and the existence of complete problems can be somewhat subtle.
In the case of \emph{syntactic} complexity classes like $\PT$ and
$\NP$, there are complete problems under what we might call
\emph{computational} reductions, even reductions in very weak
computational classes such as $\mathsf{AC}^0$.  These classes have
complete problems under \emph{logical} reductions such as first-order reductions if, and only if, there
is a logic capturing them.  In the case of $\NP$, we simply know this
to be true, but for $\PT$ it remains an open question.  In the case of
$\NPcNP$, which is a \emph{semantic} class, Gurevich already showed that the existence of a logic
implies that the class has complete problems under polynomial-time
reductions (again, we can take computational reductions in much weaker
complexity classes).  It was noted in~\cite{Dawar} that this can be
strengthened to the existence of logical reductions.  In
Section~\ref{sec:intersection}, we prove this and extend it to all
intersection classes in the polynomial hierarchy.  

This result has an interesting consequence in connection with the
graph canonical labelling problem.  It is well known that if there is
a graph canonical labelling algorithm that runs in polynomial time,
then there is a logic for $\PT$
(see~\cite[Proposition~1.7]{Dawar-thesis}).  In the case of $\NPcNP$,
we are able to show that if canonical labelling can be done in this
class, a notion we make precise below, then the existence of a logic
becomes equivalent to the question of whether the class has complete
problems under polynomial-time reductions.  For intersection classes
higher up in the polynomial hierarchy, we know that canonical
labelling can be done in the class and therefore the equivalence holds
unconditionally.  This is shown in Section~\ref{sec:reductions}.

\subsection{Logics for Intersection  Classes}\label{sec:intersection}
The following strengthening of Gurevich's result showing that if
$\NPcNP$ admits a logic capturing it, it has a complete problem under
poly-time reductions was stated in~\cite[Theorem~4]{Dawar}.  

\begin{theorem}[\cite{Dawar}]\label{thm:Dawar} $\NPcNP$ has a complete
  problem under FO reductions if, and only if, it admits a
  logic.
\end{theorem} 

We generalize this theorem to higher levels of the polynomial
hierarchy as follows. 
\begin{theorem}\label{thm:PH-intersection}
There is a $\Sigma^p_n\cap\Pi^p_n$-complete problem under first-order
reductions if, and only if, there is a logic capturing $\Sigma^p_n\cap\Pi^p_n$.
\end{theorem}
\begin{proof}
In order to prove this result, we need the following lemma:
\begin{lemma}[{\cite[p.~228]{Hodges}}]\label{lem:hodges}
Let $\sigma$ be a finite relational vocabulary. Then, there exists first-order interpretations $I_\sigma : \STRUC[\sigma]\to \STRUC[\GRAPH]$ and  $I_\sigma^{-1}$ such that 
\[ \forall \mathcal A\in \STRUC[\sigma], I_\sigma^{-1}(I_\sigma(\mathcal A))\cong \mathcal A\]
Moreover, $\forall \mathcal A,\mathcal A'\in \STRUC[\sigma], \mathcal A\cong\mathcal A' \iff I_\sigma(\mathcal A)\cong I_\sigma(\mathcal A')$
\end{lemma}

We now use this to prove Theorem~\ref{thm:PH-intersection}.

\begin{itemize}
\item[$(\Rightarrow)$] Let $Q$ be a $\Sigma^p_n\cap\Pi^p_n$-complete problem under first-order
reductions and let $\tau$ be the vocabulary of $Q$, and let
$I_{\tau}^{-1}$ be the reduction from Graphs to $\tau$-structures
given by Lemma~\ref{lem:hodges}.  We define the following logic for
any signature $\sigma$ : 
\begin{itemize}
\item $\SEN(\sigma) = \{\Theta \mid \Theta \text{ is a first-order
    interpretation from }\sigma\text{ to } \GRAPH \}$
\item $\SAT(\sigma) = \{(S,\Theta) \mid I_\tau^{-1}(\Theta(S))\in Q\}$  
\end{itemize}
This logic obviously captures $\Sigma^p_n\cap\Pi^p_n$.  This can be
seen by taking a fixed $(M,N,p)$ that witnesses the membership of $Q$
in  $\Sigma^p_n\cap\Pi^p_n$.  Then, combining this with polynomial
time machines that compute the interpretations $\Theta$ and $I_\tau^{-1}$
gives a computable map that takes $\Theta \in \SEN(\sigma)$ to a
witness $(M_\Theta,N_\Theta,p_\Theta)$ for $\MOD(\Theta) \in
\Sigma^p_n\cap\Pi^p_n$.

\item[$(\Leftarrow)$] Let $\mathcal L$ be a logic for
  $\Sigma^p_n\cap\Pi^p_n$.  Assume we have an encoding of sentences in
  $\SEN(\GRAPH)$ as integers, and let $\mathcal I$ be the the range of
  this encoding.  Let $\mathcal C$ be a deterministic Turing Machine witnessing that $\mathcal L$ captures $\Sigma^p_n\cap\Pi^p_n$ (as in Definition~\ref{def:PHcapture}).

We aim to define a class $Q$ of structures complete for graph problems
in $\Sigma^p_n\cap\Pi^p_n$ over $\tau =\langle V,E,\preceq, I\rangle$
where $V$ and $I$ are unary and $E$ and $\preceq$ are binary relation
symbols.  A structure $\mathfrak A = \langle A,V,E,\preceq,I\rangle$ belongs to Q if : 
\begin{enumerate}
\item $\preceq$ is a total, transitive, reflexive relation, i.e.\ a linear pre-order. 
\item $\forall a,b, I(a)\land I(b)\implies a\preceq b\land b\preceq a$, and $i$ is the greatest integer such that $\exists x_1,x_2\dots x_i, x_1\precnsim x_2\precnsim \dots\precnsim x_i\land I(x_i)$, where $x\precnsim y \equiv (x\preceq y\land y\npreceq x)$. In other words, $I$ picks the $i$-th equivalence class in $\preceq$
\item $\mathcal C$ on input $i$ runs in time $t\le |A|$, and outputs $(M,N,p)$
\item $|A|\ge p(|V|)$
\item $M$ accepts $\langle V,E\rangle$
\end{enumerate}
 $Q$ is in $\Sigma^p_n\cap\Pi^p_n$ : 1, 2, 3 and 4 are clearly
 computable deterministically in polynomial time. As for 5. it is both
 in $\Sigma^p_n$, by checking that there is a computation of $M$ that accepts $\langle V,E\rangle$
 in $p(|V|)$ steps, and in $\Pi^p_n$, by checking that all computations
 of $N$ of length at most $p(|V|)$ accept $\langle V,E\rangle$.
 
 To show that $Q$ is $\Sigma^p_n\cap\Pi^p_n$-hard,  let $\mathcal P$
 be a class of graphs in $\Sigma^p_n\cap\Pi^p_n$.  Let $\varphi\in
 \SEN(\GRAPH)$ be an $\mathcal L$-sentence such that $\MOD[\varphi] =
 \mathcal P$.  Let $i \in \mathcal{I}$ be the encoding of $\varphi$,
 $t$ the length of the computation of $\mathcal C$ on input $i$ and
 $(M,N,p)$ the output of the computation.  Let $k$ and $n_0$ be
 integers such that $k \geq i$, $n^k \geq t$, $n^k \geq p(n)$ for all
 $n \geq n_0$.  We describe a $k$-ary first-order interpretation $\Theta :
 \STRUC[\GRAPH]\to \STRUC[\tau]$ which is a reduction from
 $\mathcal{P}$ to $Q$ for all graphs with at least $n_0$ vertices.
 The finitely many cases of graphs with fewer than $n_0$ vertices can
 be dealt with by adding a disjunct to the formulas mapping them to
 some fixed structures inside or outside $Q$ depending on whether or
 not they are in $\mathcal{P}$ in the standard way.  Our reduction is
 given by the tuple of formulas
 $(\varphi_0,\varphi_V,\varphi_E,\varphi_\preceq,\varphi_I)$ as follows.
 \begin{itemize}
 \item $\varphi_0 \equiv\mathbf{true}$
 \item $\varphi_V(x_1,\dots, x_k) \equiv x_1 = x_2 = \dots = x_k$
 \item $\varphi_E(x_1,\dots x_k,y_1,\dots,y_k) \equiv \varphi_V(x_1,\dots,x_k) \land \varphi_V(y_1,\dots,y_k) \land E(x_1,y_1)$
 \item $\varphi_{\preceq}$ defines an arbitrary ordering of basic
   equality types of $k$-tuples from $V$.  Note that the condition
   $k\geq i$ guarantees, in particular, that there are at least $i$
   such types.
 \item $\varphi_I$ defines the $i$th equality type. $$\begin{aligned}\varphi_I(\overline{a})\equiv \exists \overline{a_1},\dots,\overline{a_{i-1}}, &\bigwedge_{1\le j<i-1} (\varphi_\preceq(\overline{a_j},\overline{a_{j+1}})\land\lnot\varphi_\preceq(\overline{a_{j+1}},\overline{a_j}))\\
&\land \varphi_\preceq(\overline{a_{i-1}},\overline{a})\land\lnot\varphi_\preceq(\overline{a},\overline{a_{i-1}})\\
 &\land \forall \overline{b}, \varphi_\preceq(\overline{b},\overline{a_i})\implies\bigvee_{1\le j< i} (\varphi_\preceq(\overline{a_j},\overline{b})\land\varphi_\preceq(\overline{b},\overline{a_j}))\end{aligned}$$
 \end{itemize}
For any graph $G$, $I(G)\in Q$ if and only if $M$ accepts $(V,E)$, if and only if $(V,E)\models \varphi$, as conditions 1, 2, 3 and 4 result from definition.
\end{itemize}
\end{proof}

\subsection{Logical and Computational  Reductions}\label{sec:reductions}
Theorem~\ref{thm:PH-intersection} has an interesting consequence.  We
know that if canonical labelling of graphs can be done in polynomial
time, then there is a logic for $\PT$.  In the case of $\NPcNP$, if
canonical labelling is in the class, we still need the additional
condition that $\NPcNP$ is a syntactic class, i.e.\ it admits complete
problems under \emph{computational} (e.g.\ polynomial-time)
reductions.  Higher up in the polynomial hierarchy, for classes
$\Sigma_n^p \cap \Pi_n^p$ where $n \geq 2$, we know that canonical
labelling is, indeed, in the class.  There the existence of a logic
becomes equivalent to the question of whether there are complete
problems under polynomial-time reductions.  To make this precise, we
first need to define what it means for canonical labelling to be in
$\NPcNP$, or $\Sigma_n^p \cap \Pi_n^p$, which are classes of decision
problems. 

An \emph{ordered graph} is a structure $(V,E,\leq)$ where $(V,E)$ is a
graph and $\leq$ is a linear order on $V$.  A \emph{canonical
  labelling function} is a function $\Can$ taking ordered graphs to
ordered graphs such that
\begin{itemize}
\item if $\Can(V,E,\leq) = (V',E',\leq')$ then $(V,E) \cong (V',E')$;
  and 
\item if $(V,E) \cong (V',E')$ then for any linear orders $\leq$ and
  $\leq'$ on $V$ and $V'$ respectively, $\Can(V,E,\leq) \cong
  \Can(V',E',\leq')$. 
\end{itemize}

We say that a canonical labelling function is in $\FP$ (the class of
function problems computable in polynomial time) if it can be computed
by a deterministic Turing machine running in polynomial time.  To
define a corresponding notion for $\NPcNP$, we use the class $\TFNP$
defined by Megiddo and Papadimitriou~\cite{Megiddo}.
\begin{definition}
  We say that a canonical labelling function $\Can$ is in $\TFNP$ if
  the graph of the function, i.e.\ $\{ (X,Y) \mid \Can(X) = Y \}$ is
  in $\PT$.
\end{definition}
As noted by Megiddo and Papadimitriou~\cite{Megiddo}, $\TFNP$ (even
though it is not a class of functions) can be understood as the
function problems corresponding to $\NPcNP$.  This allows us to prove
the following result.

\begin{theorem}\label{thm:NPcNP-canon}
If $\NPcNP$ admits a complete problem under polynomial reductions, and
there is a canonical labelling function in $\TFNP$, then $\NPcNP$
admits a complete problem under first-order reductions. 
\end{theorem}
\begin{proof}\leavevmode
If $\Can$ is in $\TFNP$, there is a nondeterministic machine
$\mathcal{G}$ which, given a string encoding an ordered graph $G$,
runs in time polynomial in the size of $G$ and each computation of
$\mathcal{G}$ either ends in rejection or, produces on the output tape
an encoding of $\Can(G)$.  Indeed, the machine $\mathcal{G}$ can
nondeterministically guess a string for $\Can(G)$, then verify that
the guess is correct and write it on the output tape or reject if it
is not. 

Let $\mathcal P$ be an $\NPcNP$-complete problem on graphs under
polynomial reductions, and $(\mathcal M, \mathcal N, p)$ be a triple witnessing this membership.

Finally, let $(M_i,p_i)_{i\in\mathcal I}$ be an enumeration of pairs
where $M_i$ is a deterministic Turing  machine with output tape and
$p_i$ is a polynomial.  We write $f_i$ for the function on strings
computed by the machine $M_i$ when clocked with the polynomial $p_i$.

We can now construct the following logic $\mathcal L$ : 
\begin{itemize}
\item $\SEN(\sigma) = \mathcal I$
\item $\SAT(\sigma)$ is the set of all $(S,i), S\in \STRUC[\sigma],
  i\in\mathcal I$ such that $\mathcal M$ accepts
  $x = f_i(\Can(I_\sigma(S)))$ in $p(|x|)$ steps.
\end{itemize}
To see that this is a logic, i.e.\ that the satisfaction relation is
well defined, let $S$ and $S'$ be two isomorphic $\sigma$-structures.
By Lemma~\ref{lem:hodges},  $I_\sigma(S)\cong I_\sigma(S')$ and
therefore $\Can(I_\sigma(S)) = \Can(I_\sigma(S'))$.  Hence, 
\[\forall \varphi\in \SEN(\sigma), S\models \varphi\iff S'\models\varphi.\]

To see that this logic captures $\NPcNP$,  let $L$ be an $\NPcNP$
decidable class of structures of signature $\sigma$.  Then,
$I_\sigma(L)$ is an $\NPcNP$ problem (as
$I_\sigma^{-1}(I_\sigma(L))=L$), so there exists $i\in\mathcal I$ such
that $M_i$ computes a reduction from $I_\sigma(L)$ to $\mathcal P$ in
time bounded by $p_i$.  
Therefore, for all $S\in \STRUC[\sigma]$, $S\in L\iff
f_i(\Can(I_\sigma(S)))\in\mathcal P$.  In other words, there is $i\in
I$ such that $\MOD[i] = L$.

Finally, note that there is a computable translation that takes us
from $i$ to a witness $(M,N,p)$ to the fact that $\MOD[i]$ is in
$\NPcNP$.  Here $M$ is the nondeterministic machine that takes as
input a $\sigma$-structure $S$ and first computes $I_\sigma(S)$.  This
can be done deterministically in polynomial time.  It then runs the
non-deterministic machine $\mathcal{G}$.  Rejecting computations of
this lead to $M$ rejecting, but accepting computations produce
$\Can(I_\sigma(S))$ on which we now run $M_i$ for
$p_i(|\Can(I_\sigma(S))|)$ steps.  Finally we run $\mathcal{M}$ on the
result.  $N$ is defined similarly except that in the last stage we
run $\mathcal{N}$.  It can now be checked that this satisfies all the
conditions for a logic capturing $\NPcNP$.  Hence by
Theorem~\ref{thm:Dawar}, there is an $\NPcNP$-complete problem under
FO-reductions. 
\end{proof}

To lift the result to higher levels of the polyomial hierarchy, we
first define what it means for graph canonical labelling to be in the
functional variant of 
$\Sigma_n^p\cap\Pi_n^p$. 
\begin{definition}
  We say that a canonical labelling function $\Can$ is in $\mathsf{F}(\Sigma_n^p\cap\Pi_n^p)$ if
  the graph of the function, i.e.\ $\{ (X,Y) \mid \Can(X) = Y \}$ is
  in $\Delta_n^p$.
\end{definition}

We can now state the following equivalence.
\begin{theorem}\label{thm:PH-canon}
For $n\ge 2$. $\Sigma_n^p\cap\Pi_n^p$ admits a complete problem under
polynomial-time reductions if, and only if, it admits a complete
problem under first-order reductions.
\end{theorem}
\begin{proof}
One implication is trivial.  For the other one, the proof is exactly
as for Theorem~\ref{thm:NPcNP-canon}, except we know that there is a
canonical labelling function in $\mathsf{F}(\Sigma_n^p\cap\Pi_n^p)$ (see~\cite{BlassGurevich-equivalence}).
\end{proof}

\section{A relativization of Gurevich's conjecture}\label{sec:relativization}
It is well-known that the conjecture of Gurevich that there is no
logic for $\PT$ implies the conjecture that $\PT$ is different from
$\NP$.  Here we show that there is a relativized world in which these
two conjectures are different, i.e.\ the first fails while the second
is true.
\begin{theorem}
There is an oracle $A$, such that there is a logic for $\PT^A$ and
$\PT^A \ne \NP^A$.
\end{theorem}
\begin{proof}
As constructed in \cite{torenvliet}, let $B$ be a set such that
$\Delta^{P,B}_2\subsetneq\Sigma^{P,B}_2$.  Then take $A$ to be a
$\Sigma^{P,B}_1$-complete set.  Then, $\PT^A$ =
$\Delta^{P,B}_2\subsetneq \Sigma^{P,B}_2 = \NP^A$.  

Moreover, since $\Delta^P_2\subset \PT^A$, there is a graph canonical
labelling function $\Can$ computable by a deterministic
polynomial-time machine with an oracle for $A$.  Let
$(M_i,p_i)_{i\in\mathcal I}$ be an enumeration of polynomial time
bounded oracle Turing Machines. We can now build a logic for $\PT^A$ : 
\begin{itemize}
\item $\SEN(\sigma) = \mathcal I$
\item $\SAT(\sigma) = \{(S,i), \Can(I_\sigma(S))\text{ is accepted by }M_i\text{ with oracle }A\} $.
\end{itemize}
\end{proof}

\section{Conclusion}
A logic capturing a complexity class requires us to find an effective
syntax for the machines that define the class \emph{and} are
isomorphism invariant.  For complexity classes that are inherently
syntactic, such as $\PT$ and $\NP$, this requirement can be met by
finding a suitable canonical labelling algorithm.  For other classes
which are inherently semantic, such as $\NPcNP$, the requirement
breaks down to finding a syntactic characterization (i.e.\ a complete
problem) in addition to a canoncial labelling algorithm.  This allows
us to explore these questions in relativized worlds.  One interesting
question to pursue would be whether the requirement for a canonical
labelling algorithm can itself be done away with in a relativized
world?  Could one devise an oracle with respect to which canonical
labelling is not in polynomial-time yet there is a logic for $\PT$?

\bibliographystyle{plain}
\bibliography{DahanDawar}
\end{document}